\documentclass[pra,twocolumn,showpacs]{revtex4}
\usepackage{amsthm}
\usepackage{amsmath}
\usepackage{latexsym}
\usepackage{amsfonts}
\usepackage{amssymb}
\usepackage{color}
\usepackage{bbm,dsfont}
\usepackage{graphicx}
\usepackage{hyperref}
\usepackage{subfigure}

\newtheorem{proposition}{Proposition}
\newtheorem{proposition?}{Proposition?}

\theoremstyle{definition}

\newcommand{\real}{\mathbb R} 
\newcommand{\half}{\frac{1}{2}} 
\newcommand{\mo}[1]{\left| #1 \right|} 

\newcommand{\hi}{\mathcal{H}} 
\newcommand{\ip}[2]{\left\langle\,#1\,|\,#2\,\right\rangle} 
\newcommand{\no}[1]{\left\|#1\right\|} 
\newcommand{\tr}[1]{\textrm{tr}\left[#1\right]} 
\newcommand{\rank}{\mathrm{rank}\,} 

\newcommand{\id}{\mathbbm{1}} 

\newcommand{\ltwo}[1]{L^2(#1)} 

\newcommand{\G}{\mathsf{G}}
\renewcommand{\P}{\mathsf{P}}
\newcommand{\M}{\mathsf{M}}
\newcommand{\N}{\mathsf{N}}

\begin{document}
\title[]{Noise Robustness of the Incompatibility of Quantum Measurements}

\begin{abstract}
The existence of \emph{incompatible measurements} is a fundamental phenomenon having no explanation in classical physics. Intuitively, one considers given measurements to be incompatible within a framework of a physical theory, if their simultaneous implementation on a single physical device is prohibited by the theory itself. In the mathematical language of quantum theory, measurements are described by POVMs (positive operator valued measures), and given POVMs are by definition incompatible if they cannot be obtained via coarse-graining from a single common POVM;
this notion generalizes noncommutativity of projective measurements. In quantum theory, incompatibility can be regarded as a resource necessary for manifesting phenomena such as Clauser-Horne-Shimony-Holt (CHSH) Bell inequality violations or Einstein-Podolsky-Rosen (EPR) steering which do not have classical explanation. We define operational ways of \emph{quantifying} this resource via the amount of added classical noise needed to render the measurements compatible, i.e., useless as a resource. In analogy to entanglement measures, we generalize this idea by introducing the concept of incompatibility measure, which is monotone in local operations. In this paper, we restrict our consideration to binary measurements, which are already sufficient to explicitly demonstrate nontrivial features of the theory. In particular, we construct a family of incompatibility monotones operationally quantifying violations of certain scaled versions of the CHSH Bell inequality, prove that they can be computed via a semidefinite program, and show how the noise-based quantities arise as special cases. We also determine maximal violations of the new inequalities, demonstrating how Tsirelson's bound appears as a special case. The resource aspect is further motivated by simple quantum protocols where our incompatibility monotones appear as relevant figures of merit.
\end{abstract}

\author{Teiko Heinosaari}
\affiliation{Turku Centre for Quantum Physics, Department of Physics and Astronomy, University of Turku, Finland}
\author{Jukka Kiukas}
\affiliation{Department of Mathematics, Aberystwyth University, Aberystwyth, SY23 3BZ, United Kingdom}
\author{Daniel Reitzner}
\affiliation{Institute of Physics, Slovak Academy of Sciences, D\'ubravsk\'a cesta 9, 845 11 Bratislava, Slovakia}


\maketitle

\section{Introduction}\label{sec:intro}

Small-scale physical systems exhibit features that cannot be explained by classical physics. This is often considered as a resource in the context of quantum information theory: it allows one to perform useful tasks (e.g., computational ones) not implementable via classical protocols, and it is costly to maintain in practice. The precise content of this idea has been mathematically formulated for \emph{quantum states}, in terms of entanglement \cite{Horo09} and also more generally in \cite{CoFrSp14,BrGo15}. Without going into details of such \emph{resource theories}, we list their main ingredients motivating our study: (A) the definition of the ``void resource'', (B) specification of quantum operations $\Lambda$ \footnote{A quantum operation is a (completely) positive transformation on the set of quantum states.} that cannot transform a void resource into a useful one, and (C) quantification of the resource. For instance, in the case of entanglement, void resources are the separable states, the operations in (B) are local operations and classical communication (LOCC), and (C) refers to various entanglement measures.

As demonstrated by the existence of a multitude of different entanglement measures, there is in general no unique way of quantifying quantum resources. However, the idea that noise cannot create such a resource suggests an appealing and operational way of quantifying it via \emph{the least amount of classical noise needed to render the resource void} \cite{ViTa99,GrPoWi05}. Such a quantity is often referred to as \emph{robustness} \cite{BrGo15,ViTa99}. In this context, the mathematical description of noise is convex-geometric: as a simple scheme,
consider a preparation device which outputs a bipartite state $\rho$ (a density matrix) with probability $1-\lambda$ and a completely mixed state with probability $\lambda$; the resulting state is then the convex mixture 
$\rho_\lambda =(1-\lambda)\rho+\lambda \id/d^2$,
where $d$ is the dimension of each subsystem. This noise model has been used in investigating non-classicality of bipartite correlations. In fact, for the noise parameter $\lambda$ large enough, the state $\rho_\lambda$ has a local classical model so that all Bell inequalities hold; this goes back to the construction of the Werner states \cite{We89}, and the same construction works also for the arbitrary pure entangled state $\rho$ \cite{AlPiBaToAc07}.

Having briefly reviewed relevant existing ideas concerning quantum states as resources, we now change the point of view, considering \emph{quantum measurements as resources} instead. In order to motivate this, we note that characterizing the nonclassicality of a quantum system as a property of \emph{state alone} has a limited practical significance as the set of available measurements is almost always restricted in real experiments.

In particular, obtaining violations of a Bell inequality fails if measurements are not appropriately chosen, even if the quantum state is maximally entangled. The basic scenario we have in mind here is typical in quantum information theory: two local parties, Alice and Bob, share a bipartite state, and are capable of performing some restricted set of \emph{local} measurements in their respective laboratories. Thus, in this context the state is a \emph{nonlocal resource}, while Alice's and Bob's measurements represent \emph{local resources}. In addition to Bell inequality violations, the chosen setting is relevant for, e.g., Einstein-Podolsky-Rosen (EPR) steering \cite{UoMoGu14}.

The task of this paper is to systematically study a specific measurement resource, \emph{incompatibility} of the measurements, by appropriately formulating the above points (A)--(C) for measurements in Sec.~\ref{sec:incompatibilitymeasurements}. Explicitly, this consists of (A) stating the mathematical definition of compatibility (defining the void resource), (B) observing that incompatibility cannot be created by any quantum operation, and (C) introducing the concept of an \emph{incompatibility monotone} for the quantification of the resource.

The main part of the paper is devoted to constructing and studying concrete incompatibility monotones for binary measurements. We begin by constructing a family of incompatibility monotones via a semidefinite program in Sec.~\ref{sec:Ia}, also connecting them explicitly to the Clauser-Horne-Shimony-Holt (CHSH) Bell inequality setting, where incompatibility appears manifestly as a local resource. In Sec.~\ref{sec:noiserob} we then show that these incompatibility monotones can be operationally interpreted as measures of noise robustness with an adaptation of the noise model described above.
In Sec.~\ref{sec:maxinc} we use the developed formalism to define maximal incompatibility, determine which observables have this property, and give an example of how they could be constructed with restricted (concrete) experimental resources. Finally, in Sec.~\ref{sec:game}, we demonstrate the usefulness of the results in the form of a ``quantum game.''
\section{Incompatibility as a resource in quantum theory}
\label{sec:incompatibilitymeasurements}
In this section we consider incompatibility from the general resource theoretic point of view, as outlined in the introduction. In particular, we introduce the notion of ``incompatibility monotone'' for quantification of the resource. Here we wish to point out that after the appearance of the original eprint version of the present paper, similar general ideas appeared in \cite{pusey}. Here we only consider the general aspects to the extent they apply to binary measurements the detailed analysis of which is our main topic. Full resource theory of incompatibility is beyond the scope of this paper.

\subsection{Definition of incompatibility}

A measurement (or observable) in any probabilistic theory is represented by a map that associates a probability distribution (describing the measurement outcome statistics) to any state of the system described by the theory. \emph{Incompatibility of measurements} refers to a phenomenon occurring in quantum theory and general probabilistic theories: there are observables that do not have a common refinement from which they could be obtained by coarse-graining the associated probability distributions. This formal concept corresponds to the intuitive idea of measurements that cannot (even in principle) be performed simultaneously by one device. We now proceed to formulate it precisely in the quantum case.

A quantum measurement (with discrete outcomes $i=1,\ldots, m$) is described by a map $\rho\mapsto ({\rm tr}[\rho M_i])_{i=1}^m$, where $\rho$ is a density matrix, and the associated probability distribution $({\rm tr}[\rho M_i])_{i=1}^m$ is determined by a positive operator valued measure (POVM) $\M=(M_i)_{i=1}^m$, satisfying $0\leq M_i\leq \id$ and $\sum_i M_i=\id$. In general, operators $M$ satisfying $0\leq M\leq \id$ are called \emph{effects}.

We now come to an essential definition: two measurements $\M=(M_{i})$ and $\N=(N_{j})$ are said to be \emph{compatible} (or \emph{jointly measurable}) \cite{Lu85} if there exists a common refinement, i.e., a POVM $\G=(G_{ij})$ such that \begin{align*}
M_i &=\sum_j G_{ij}, & N_j &= \sum_i G_{ij}.
\end{align*}
If $\M$ and $\N$ are not compatible, they are said to be \emph{incompatible}. A measurement $\M=(M_{i})$ is projective (or a von Neumann measurement), if $M_i$ is an orthogonal projection for all $i$'s. It is well known that two projective measurements $\M$ and $\N$ are compatible exactly when they are commutative, i.e., $[M_i,N_j]=0$ for all $i,j$. In general, commutativity is not necessary for compatibility; see, e.g., \cite{LaPu97} for a discussion.

\subsection{Incompatibility under quantum operations}

Suppose we measure an observable $\M$ after first applying a quantum channel to an initial state $\rho$. Then the probability of getting outcome $i$ is ${\rm tr}[\Lambda(M_i)\rho]$, where $\Lambda$ is the channel in the Heisenberg picture. Hence, we can interpret this as a measurement of the transformed POVM $\Lambda(\M):=(\Lambda(M_i))$. Since $\Lambda$ is by assumption (completely) positive and unital, this indeed defines a valid POVM.

In order to consider incompatibility as a resource in the sense outlined in the introduction, we should specify quantum channels $\Lambda$ that cannot create incompatibility. For entanglement, which is a nonlocal resource, the relevant operations are LOCC (local operations and classical communication). In contrast, nonlocality does not play any role in the definition of incompatibility, and accordingly, there is no restriction on the set of operations; in fact, \emph{incompatibility cannot be created by any quantum channel}. In order to prove this, suppose that the measurements $\M$ and $\N$ are compatible with a joint POVM $\G$, and let $\Lambda$ be a quantum channel. Then by positivity, $\Lambda(\G)$ is a POVM, and by linearity, we have $\sum_j \Lambda(G_{ij})=\Lambda(M_i)$ and $\sum_i \Lambda(G_{ij})=\Lambda(N_j)$, i.e., $\Lambda(\G)$ is a joint POVM for the transformed POVMs $\Lambda(\M)=(\Lambda(M_i))$ and $\Lambda(\N)=(\Lambda(N_i))$. Hence, $\Lambda(\M)$ and $\Lambda(\N)$ are compatible.

\subsection{Incompatibility monotones}

The observation of the previous subsection now suggests the following general requirements that should be satisfied by any \emph{quantification} of incompatibility.

Let $\mathcal M$ be some collection of POVMs. A real valued function $I$ defined on pairs of POVMs from $\mathcal M$ is called an \emph{incompatibility monotone} on $\mathcal M$ if it fulfills the following conditions:
\begin{itemize}
\item[(i)] $I(\M,\N)=0$ if and only if $\M$ and $\N$ are compatible.
\item[(ii)] $I(\M,\N)=I(\N,\M)$.
\item[(iii)] $I(\Lambda(\M),\Lambda(\N))\leq I(\M,\N)$ for all POVMs $\M,\N$ and any quantum channel $\Lambda$.
\end{itemize}
Here property (i) is the basic requirement: the quantification must distinguish void resources. Property (ii) is natural because incompatibility is related to \emph{pairs} of measurements. Property (iii), \emph{monotonicity} in quantum operations, corresponds exactly to monotonicity of entanglement measures in LOCC operations \cite{Horo09}. As one consequence of the definition, any incompatibility monotone decays as expected under Markovian dynamical evolution. Note also that (iii) implies \emph{unitary invariance}, i.e.,
\begin{equation*}
I(U\M U^*,U\N U^*)=I(\M,\N)
\end{equation*}
for all unitaries $U$ on $\hi$.

\subsection{Application: CHSH violation in the binary case}

We now proceed to restrict the setting to binary measurements; these extract exactly one bit of classical information from a given quantum state. Since this restricted setting will be used for the rest of the paper, it deserves detailed discussion. Note that we do \emph{not} restrict the dimension of the Hilbert space.

Binary measurements are represented by two-element POVMs $\M=(M, \id-M)$, corresponding to the outcomes $1$ and $0$, respectively.
Two binary observables $\M$ and $\N$ are incompatible (in the sense of the definition above) exactly when it is \emph{not} possible to construct a joint measurement that outputs two bits of classical information in such a way that the first bit represents the measurement outcome of $\M$ and the second the measurement outcome of $\N$. If such a joint measurement $\G$ exists (i.e., $\M$ and $\N$ are compatible), then $\G$ is a four-element POVM $\G=(G_{11},G_{10},G_{01},G_{00})$ satisfying
\begin{align}
G_{11}+G_{10}&= M, & G_{11}+G_{01}&=N, & \sum_{ij} G_{ij}&=\id;\label{equalityconst}\\
G_{ij}&\geq 0, & i,j=0,1.\label{semidefconst}
\end{align}
As an illustration, one can think of $\G$ as a measurement device with four LEDs; two of the LEDs correspond to the measurement outcome $1$ for $\M$, and similarly for $\N$ (see Fig.~\ref{fig:partitions}).

It is important to emphasise that there is in general no explicit analytical expression of $\G$ in terms of $\N$ and $\M$; this makes it nontrivial to decide whether two given measurements are incompatible, even in a one-qubit system.

\begin{figure}
\begin{center}
\includegraphics[scale=0.6]{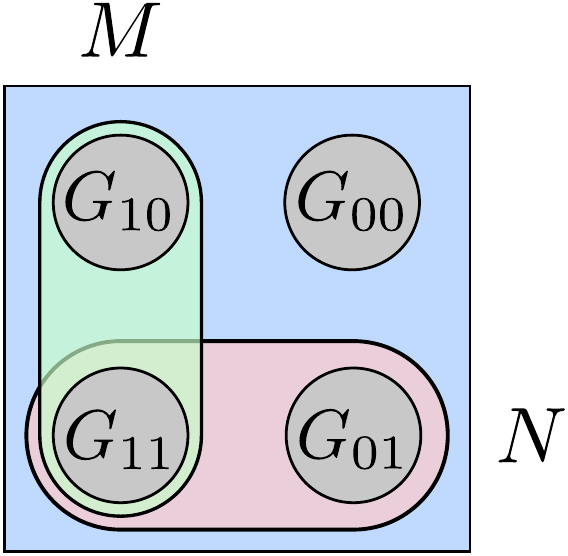}
\end{center}
\caption{\label{fig:partitions} Joint measurement $\G$ of two binary observables $\M=(M,\id-M)$ and $\N=(N,\id-N)$ is a four-element POVM such that (\ref{equalityconst}) and (\ref{semidefconst}) hold --- the outcome of the measurement $\M$ corresponds to the first outcome (index) of measurement $\G$ and the outcome of the measurement $\N$ corresponds to the second outcome (index).}
\end{figure}

Perhaps the most important task where incompatibility concretely appears as a resource, is obtaining CHSH Bell inequality violations. For the benefit of readers not familiar with CHSH, we briefly describe the setup with notations to be used below: two (spatially separated) parties, Alice and Bob, both have a pair ($\M^A,\N^A$), ($\M^B, \N^B$) of of \emph{local} binary measurements. For each run of the experiment, Alice and Bob perform \emph{one} measurement each, and record the results. After sufficiently many repetitions, they compare the results to construct a joint correlation table. The correlations are described by the expectation value operators $A_1=2M^A-\id$, $A_2=2N^A-\id$, and $B_1$, $B_2$ similarly for Bob. The \emph{CHSH Bell inequality} for given bipartite state $\rho$ is then given by ${\rm tr}[\rho \mathbb B]\leq 1$, where
\begin{equation}\label{bellop}
\mathbb B=\frac{1}{2}\left[A_1\otimes(B_1+B_2)+A_2\otimes(B_1-B_2)\right].
\end{equation}
Violation of the inequality shows that the correlations cannot be described by a local classical model; for more information we refer to \cite{WeWo01}.

The following result, proved recently in \cite{WoPeFe09}, shows that incompatibility is a resource in this context: a given pair of Alice's measurements is incompatible if and \emph{only if} there exists a bipartite state and a pair of measurements for Bob, such that the CHSH Bell inequality is violated. The amount of Bell inequality violation can be expressed as a semidefinite program, the dual of which quantifies deviation from compatibility in a certain sense; we generalize this idea below.

As an example of a more general scenario (where measurements can have multiple outcomes), incompatibility is crucial for EPR steering \cite{UoMoGu14,QuVeBr14}.

\section{Incompatibility monotones via semidefinite program}
\label{sec:Ia}

The rest of the paper is devoted to constructing instances of incompatibility monotones in the binary setting. In this section, we are motivated by the following fact: it is important to have \emph{efficiently computable} quantifications of quantum resources. A convex optimization problem constrained by semidefinite matrix inequalities is called a \emph{semidefinite program} (SDP) \cite{VaBo96}; they are efficiently computable, and appear frequently in quantum information theory \cite{DoPaSp03,JaMi05,ViWe02}. In fact, entanglement measures are often defined via suitable optimization \cite{Horo09}.

We now demonstrate that a large class of incompatibility monotones may be computed via SDP. The construction is based on the convex structure of the set of joint POVMs for a compatible pair of  POVMs. In fact, given two binary POVMs $\M$ and $\N$, the possible joint POVMs $\G=(G_{ij})_{i,j=0,1}$ are specified by the \emph{equality constraints} \eqref{equalityconst}, together with the \emph{semidefinite constraints} \eqref{semidefconst}. Hence, deciding whether two binary measurements are incompatible is manifestly an SDP; this fact was pointed out in \cite{WoPeFe09}.

Here we develop the program further by observing that it can be made \emph{feasible} \cite{VaBo96} by deforming the semidefinite constraints. 
The most general \emph{linear symmetric deformation by identity} is given by a real symmetric $2\times2$ matrix ${\bf a}=(a_{ij})$ with positive elements: we replace \eqref{semidefconst} by
\begin{equation}\label{semidefmu}
G_{ij}+\mu a_{ij}\id\geq 0,
\end{equation}
where $\mu\geq 0$ affects the deformation. The semidefinite program is now as follows:
\begin{align}
&\text{Minimize $\mu\geq 0$ over all operators $G_{ij}$}\label{program}\\
& \text{satisfying the constraints \eqref{equalityconst} and \eqref{semidefmu}.}\nonumber
 \end{align}
For ${\bf a}=0$ this reduces to the original decision problem of whether $\M$ and $\N$ are incompatible. It is easy to see that if $a_{ij}>0$ for at least one $(i,j)$, the program is feasible, i.e., for some $\mu\geq 0$ there exist four matrices $G_{ij}$ satisfying the constraints. We call such a $\mu$ \emph{admissible} [for the pair $(\M,\N)$ and a matrix ${\bf a}$]. We let $I_{\bf a}(\M,\N)$ denote the associated minimum value of $\mu$. It is clear that we can use the equality constraints to parametrize the four matrices $G_{ij}$ in terms of a single matrix, and perform the optimization over that.
The following proposition lists the main properties of $I_{\bf a}$; the proof is given in Appendix~\ref{app:proof}.

\begin{proposition}\label{incomp_prop}
$I_{\bf a}$ is an incompatibility monotone for each symmetric matrix ${\bf a}$ with positive elements. In addition, it has the following properties:
\begin{itemize} 
\item[(a)] If $M=\sum_n t_n M_n$ where $t_n>0$ and $\sum_n t_n=1$, then $I_{\bf a}(\M,\N)\leq \max_n I_{\bf a}(\M_n,\N)$.
\item[(b)] If $\mathcal H=\oplus_{n} \mathcal H_n$ and $\M=\oplus_n \M_n$, $\N=\oplus_n \N_n$, then $I_{\bf a}(\M,\N)=\max_n I_{\bf a}(\M_n,\N_n)$.
\item[(c)] If $V:\mathcal K\to \mathcal H$ is an isometric embedding of a Hilbert space $\mathcal K$ into $\mathcal H$, then $I_{\bf a}(V^*\M V,V^*\N V)=I_{\bf a}(\M,\N)$.
\item[(d)] $I_{\bf a}(\id-\M,\id-\N)=I_{\bf \tilde a}(\M,\N)$ where $\tilde a_{ij}=a_{i\oplus 1,j\oplus 1}$ with $\oplus$ being binary addition.
\item[(e)] Suppose that $\M$ and $\N$ are projective measurements. Let $\Theta$ be the set of angles $0<\theta<\pi$ for which $\frac 12(1+\cos\theta)$ belongs to the spectrum of the operator
$$\id-(M-N)^2.$$ (Note that eigenvalues $0$ and $1$ are excluded.) Then
$$
I_{\bf a}(M,N)=\sup_{\theta\in \Theta} I_{\bf a}(\P_0,\P_{\theta}),
$$
where $P_\theta=\frac 12(\id+\sin\theta\sigma_x+\cos\theta\sigma_z)$.
\end{itemize}
\end{proposition}
Part (e) shows that for projections the calculation of $I_{\bf a}$ reduces to diagonalizing the operator $\id-(M-N)^2$, which is the central element of the algebra generated by the two projections \cite{Ha69,RaSi89}. 
Its spectrum (excluding $0$ and $1$) equals that of $MNM$ and $NMN$, which are often easier to diagonalize.

Interestingly, it turns out that incompatibility measures defined by the above SDP can always be expressed in terms of operational quantities related to a correlation experiments in the standard CHSH setup. 
Since the identity operator always satisfies the conditions \eqref{equalityconst} and \eqref{semidefmu} for large enough $\mu$, the program is \emph{strictly feasible}, and consequently, \emph{strong duality}  holds, i.e., $I_{\bf a}(\M,\N)$ coincides with the value given by the associated \emph{dual program} \cite{VaBo96}. 
The dual program can be written in terms of the CHSH quantities following the method of \cite{WoPeFe09}, where a special case was considered. 
We postpone the details of computation to Appendix~\ref{app:sdp}. 
The result is a scaled version of the CHSH inequality, 
\begin{equation}\label{dual4}
I_{\bf a}(\M,\N)=
\sup_{\psi,B_1,B_2}\frac{ \langle \psi | \half(\mathbb{B}-\id)\psi\rangle}{\langle \psi | (\id\otimes S_{\bf a})\psi\rangle},
\end{equation}
where the supremum is over all $\|\psi\|=1,\,-\id\leq B_1,B_2\leq \id$, the Bell operator $\mathbb B$ is defined in Eq.~\eqref{bellop} with $A_1=\id-2N$, $A_2=2M-\id$, and we have denoted
\begin{align*}
S_{\bf a}&=\half[ 
a_{00}(\id-B_2)+a_{11}(\id+B_2) +2a_{01}\id].
\end{align*}
Note that $S_{\bf a}$ depends only on Bob's measurements. We observe that $I_{\bf a}(\M,\N)=0$ (i.e., $\M$ and $\N$ are compatible) if and only if CHSH Bell inequality is not violated. The special case considered in \cite{WoPeFe09} is given by ${\bf a}=\frac 12 \id_2$, in which case $S_{\bf a}=\frac 12\id$, so that \eqref{dual4} exactly gives the maximum possible violation of the Bell inequality with Alice's measurements fixed to be $\M$ and $-\N$.

\section{Incompatibility monotones quantifying noise-robustness}
\label{sec:noiserob}
The monotones defined in the preceding section only have an operational meaning in the context of the CHSH Bell scenario, where incompatibility appears as (Alice's) local resource. We now proceed to define monotones with a \emph{direct} operational interpretation completely analogous to the \emph{noise-robustness} idea for quantum states discussed in the Introduction. This interpretation is independent of the Bell scenario. However, as we will see in the next section, these monotones actually turn out to be special cases of the SDP monotones of the preceding section.

\subsection{A simple model for classical noise}

We begin by the description of an \emph{addition of classical noise}, in the sense of random fluctuations on measurement devices, in analogy to preparation of states as discussed in the introduction: we deform a POVM $(M_i)$ into $\M_{\lambda,p}=(M_{\lambda,p;i})$, where
\begin{equation}\label{eq:mixing}
M_{\lambda,p;i} =(1-\lambda)M_i+\lambda p_i\id \, ,
\end{equation}
 where $(p_i)$ is a probability distribution and $0<\lambda<1$. This can be understood as follows: in each run of the experiment, the original quantum measurement $\M$ will only be realized with probability $1-\lambda$; otherwise the device just draws the outcome randomly from the fixed probability distribution $(p_i)$. Hence $\lambda$ describes the magnitude of the classical noise, and $p$ is its distribution.
 

For a binary POVM $\M=(M, \id-M)$, it is convenient to write $p_1=\tfrac 12 (1+b)$, $p_0=\tfrac 12(1-b)$ in terms of the bias parameter $b\in [-1,1]$; accordingly, we denote $\M_{\lambda,b}=\M_{\lambda,p}$ in this case. Then $\M_{\lambda,b}=(M_{\lambda,b}, \id-M_{\lambda,b})$ where
\begin{align}\label{eq:noisym}
M_{\lambda,b}&= (1-\lambda) M+\lambda \tfrac 12(1+b) \id.
\end{align}

\subsection{Quantifying incompatibility via noise robustness}

Since the observable $\M_ {\lambda,b}$ gets closer to a trivial POVM as $\lambda$ increases, any initially incompatible pair of binary POVMs $\M$ and $\N$ is expected to become compatible when both are modified according to \eqref{eq:noisym}, at some value of $\lambda$. 
Hence, the number
\begin{equation}\label{noiseI}
I^{\rm noise}_{b}(\M,\N):=\inf \{ \lambda\geq 0\mid \M_{\lambda,b},\N_{\lambda,b}\text{ compatible}\}
\end{equation}
provides an operational quantification of quantum incompatibility of a pair $(\M,\N)$.
Using arguments similar to ones in \cite{BuHeScSt13}, it is straightforward to show that this number is at most $1/2$. The $b$-optimized quantity $j(\M,\N):=1-\inf_b I^{\rm noise}_b(\M,\N)$  has been referred to as the \emph{joint measurability degree} of the POVMs $\M$ and $\N$ \cite{BuHeScSt13,HeScToZi14}; our specialty here is to investigate the case of fixed bias $b$.

For the sake of comparison, let us briefly consider another simple noise model where the measurements are modified by a quantum channel. This was used in \cite{UoMoGu14} to investigate EPR steerability.
Let $\Lambda_0(\cdot)={\rm tr}[(\cdot)]\id/d$ denote the completely depolarizing channel in dimension $d<\infty$, and set 
\begin{align*}
I_{\rm steer}(\M,\N):=\inf \{\lambda\geq 0\mid & \,(1-\lambda)\M +\lambda \Lambda_0(\M) \text{ and}\\
&\,(1-\lambda)\N +\lambda \Lambda_0(\N) \text{ compatible}\},
\end{align*}
for any binary measurements $\M$, $\N$. 
It is important to note that $I_{\rm steer}(\M,\N)$ is strictly smaller than $1$, at least for projective measurements $\N,\M$. In fact, it follows from the results of \cite{UoMoGu14,WiJoDo07} that $I_{\rm steer}(\M,\N)\leq (d-\sum_{n=1}^d \tfrac 1n)/(d-1)$.
From the results of \cite{UoMoGu14} it is furthermore clear that the quantity $I_{\rm steer}(\M,\N)$ can be interpreted as an \emph{operational quantification of incompatibility as a steering resource}, in the following sense: it is the maximal amount of noise that can be added to the maximally entangled state so that the resulting state is still steerable with Alice's measurements $\M$ and $\N$.

We also note that related convex-geometric robustness measures for incompatibility and steering appeared recently in \cite{erkka} and \cite{Sk14}.

The following result connects the noise-robustness approach to the general resource-theoretic ideas discussed above.
\begin{proposition} 
The functions $I_b^{\rm noise}$ and $I_{\rm steer}$ are incompatibility monotones.
\end{proposition}
\begin{proof}
It is clear from the definitions that both $I_b^{\rm noise}$ and $I_{\rm steer}$ satisfy (i) and (ii).
To prove (iii), it is enough to make the following observation.
If two POVMs $\M$ and $\N$ are compatible, then $\Lambda(\M)$ and $\Lambda(\N)$ are compatible for any channel $\Lambda$.
Namely, if $\G$ is a joint POVM of $\M$ and $\N$, then $\Lambda(\G)$ is a joint POVM of $\Lambda(\M)$ and $\Lambda(\N)$. 
Since $\Lambda(\M_{\lambda,b})=\Lambda(\M)_{\lambda,b}$ and $\Lambda\circ \Lambda_0=\Lambda_0$ for any channel $\Lambda$, the monotonicity (iii) holds for $I_b^{\rm noise}$ and $I_{\rm steer}$, respectively.
\end{proof}

In the rest of the paper, we concentrate on the monotones $I_b^{\rm noise}$. Further properties of the steering incompatibility monotone $I_{\rm steer}$ will be investigated in a separate publication. 

\section{Analysis of noise-based monotones}

In this section, we analyze the monotones $I_b^{\rm noise}$ of the preceding section in detail. We begin by showing that they are essentially equivalent to the SDP-computable monotones introduced in Sec. \ref{sec:Ia}. In fact, note first that if $I$ is an incompatibility monotone and $f:[0,\infty)\to [0,\infty)$ is a strictly increasing function with $f(0)=0$, then the composite function $f\circ I$ is also an incompatibility monotone. The following proposition shows that every SDP-computable monotone $I_{\bf a}$ reduces to $I_b^{\rm noise}$ in this way. 
\begin{proposition}\label{propnoise}
Fix a symmetric matrix ${\bf a}$, denote $a=\sum_{i,j}a_{ij}$, and $f_{a}(\mu)=a\mu/(1+ a\mu)$ for all $\mu\geq 0$. Then $$I^{\rm noise}_b=f_{a}\circ I_{\bf a},$$
where $b=2(a_{11}+a_{01})/a-1$.
\end{proposition}
\begin{proof} For each $\mu\geq 0$, define a one-to-one map $(G_{ij})\mapsto (\tilde G_{ij})$ between four-tuples of operators via
\begin{align*}
\tilde G_{ij} &=(1+\mu a)G_{ij}-\mu a_{ij}\id.
\end{align*}
Putting then $\lambda:=f_{a}(\mu)$ we see that $(G_{ij})$ satisfies \eqref{semidefconst} and \eqref{equalityconst} for the pair $(\M_{\lambda, b},\N_{\lambda, b})$, if and only if $(\tilde G_{ij})$ satisfies \eqref{semidefmu} and \eqref{equalityconst} for $(\N,\M)$. From this the claim follows immediately.
\end{proof}

The following analogy to entanglement quantification is worth noting at this point: for a given state $\rho$, and a fixed separable state $\rho_0$, the authors of \cite{ViTa99} call the minimum value of $\mu$ for which $[1-f_{a=1}(\mu)]\rho+f_{a=1}(\mu) \rho_0$ is entangled \emph{the robustness of $\rho$ relative to $\rho_0$}.

Due to the above proposition, study of the SDP-computable monotones $I_{\bf a}$ can, without loss of generality, be restricted to the special case where ${\bf a}$ is diagonal, with $a_{11}=\tfrac 12(1+b)$ and $a_{00}=\tfrac12(1-b)$, and done using the noise-based monotone $I_{b}^{\rm noise}$. Furthermore, Proposition~\ref{incomp_prop} (e) shows that in the case where $M,N$ are projections, the quantity $I_{b}^{\rm noise}(\M,\N)$ is completely determined by the spectrum of $\id-(M-N)^2$, and the special values $I_{b}^{\rm noise}(\P_0,\P_\theta)$.

Accordingly, we now proceed to investigate these values in detail. The restriction to projections is to some extent justified by the intuition that projections represent \emph{sharp} quantum measurements with no intrinsic noise. (This terminology can be made precise in various ways, see e.g.~\cite{BuLaMi96}.) We can find $I_b^{\rm noise}(\P_0,\P_\theta)$ using the definition \eqref{noiseI}, together with the known characterization of compatibility of binary qubit measurements \cite{Busch}; see also \cite{TeikoDaniel,YLLO}, and a generalization \cite{us} by the authors of the present paper. The result is as follows (see Appendix~\ref{app:qubit} for derivation): $\lambda=I_b^{\rm noise}(\P_0,\P_\theta)$ is the unique solution $0\leq \lambda \leq 1/2$ of the equation
\begin{equation}\label{comp_eq}
[(1-\lambda)^2\cos\theta-\lambda^2b^2]^2=2(1-\lambda)^2-1+2\lambda^2b^2.
\end{equation}
Representative solutions are plotted in Fig.~\ref{fig:theta}.
\begin{figure}
\begin{center}
\includegraphics[scale=0.8]{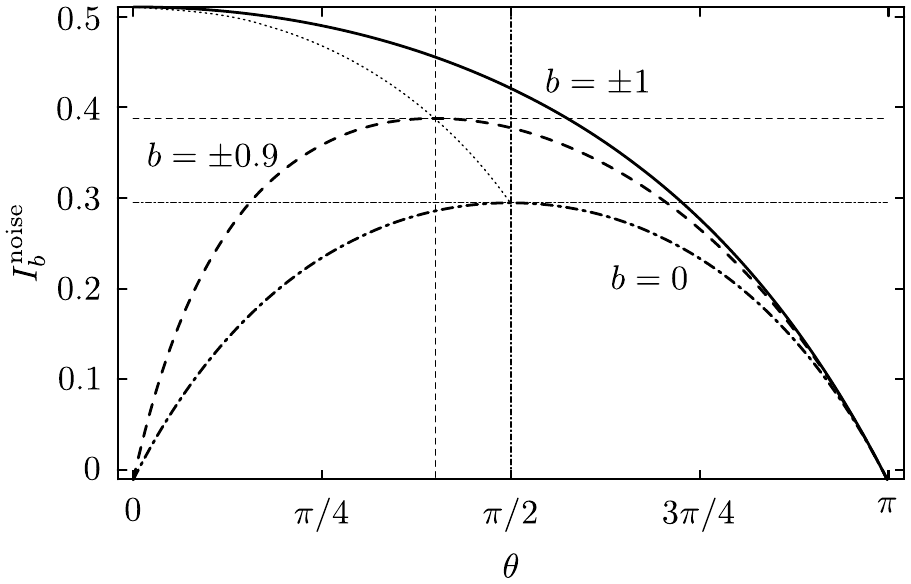}
\end{center}
\caption{\label{fig:theta}The dependence of $I^{\rm noise}_b(\P_0,\P_\theta)$ on the angle $\theta$ between the Bloch vectors of the two qubit projections for various choices of $b$. For $b=\pm1$ we can observe discontinuity at $\theta=0$, where $I^{\rm noise}_b=0$. The dotted line depicts the set of maxima for different choices of $b$, showing that for $\theta>\pi/2$ the incompatibility never reaches the largest value for given $b$. }
\end{figure}
Two interesting special cases, namely $b=0$ and $b=\pm 1$ can be solved analytically for each $\theta$; we get
\begin{align*}
I^{\rm noise}_{b=0}(\P_0,\P_\theta)&=1-(1+\sin\theta\,)^{-\frac 12},\\
I^{\rm noise}_{b=\pm 1}(\P_0,\P_\theta) &= 1-\left[1+ \sqrt{(1+\cos\theta)/2}\,\right]^{-1}.
\end{align*}
For general projective measurements, the value $I^{\rm noise}_b(\M,\N)$ is obtained by maximising over those $\theta$ for which $\tfrac 12(1+\cos\theta)$ is an eigenvalue of $\id-(M-N)^2$; see Fig.~\ref{fig:qubit} below. Concerning the above special cases, let us first take the unbiased one for $b=0$, which by \eqref{dual4} gives exactly the maximum possible CHSH violation. We get
\begin{equation}\label{Inolla}
I^{\rm noise}_{b=0}(\M,\N) =1-(1+2\no{[M,N]})^{-\frac 12},
\end{equation}
due to the fact that $\no{[M,N]}=\sup_{\theta\in \Theta}\no{[P_0,P_\theta]}$. Hence, $I^{\rm noise}_0(\M,\N)$ is a function of the commutator of the projections, as expected from the known properties of the CHSH operator $\mathbb B$ (see, e.g., \cite{KiWe10}). At the other extreme, the maximally biased case $I^{\rm noise}_{b=\pm 1}(\P_0,\P_\theta)$ is an increasing function of the eigenvalues of $MNM$ (excluding $0$ and $1$), and hence
\begin{equation}\label{Iminus}
I^{\rm noise}_{b=\pm 1}(\M,\N) =\begin{cases} 1-(1+\no{MNM})^{-1}, & [M,N]\neq 0,\\ 0, & [M,N]=0.\end{cases}  
\end{equation}
Note that in the special case where the projections commute, the spectrum of $MNM$ only has values $0$ and $1$, which are excluded, hence the discontinuity (see Fig.~\ref{fig:theta}).

Another important aspect is the apparent monotonicity of $I^{\rm noise}_b(\P_0,\P_\theta)$ in $|b|$ for $\theta\in(0;\pi)$; see Figs.~\ref{fig:theta} and \ref{fig:qubit}. That this indeed holds for all projective measurements $\M$ and $\N$, i.e., for $I_b^{\rm noise}(\M,\N)$, is proved in Appendix~\ref{app:monotonicity}. This shows that \emph{the noise robustness of incompatibility of any pair of measurements increases when the noise is biased}. Interestingly, as we see from Fig.~\ref{fig:theta}, this effect becomes dramatic when the measurements are close to commutative; the difference is best reflected in the extreme cases \eqref{Inolla} and \eqref{Iminus} which differ maximally (i.e., by $1/2$) at the commutative limit.

\section{Maximal incompatibility}
\label{sec:maxinc}

Having investigated the detailed structure of the incompatibility monotones $I^{\rm noise}_{b}$, it is now natural to ask which pairs of effects are \emph{maximally incompatible} in this sense. From the quantum resource point of view, it corresponds to the following question: \emph{which pairs of binary quantum measurements are most robust against noise?}

\subsection{Generalized Tsirelson bound}

We proceed to derive maximal incompatibility for all the concrete monotones considered above. Given any SDP-computable monotone $I_{\bf a}$ and the corresponding noise-based one $I_b^{\rm noise}=f_a\circ I_{\bf a}$ as in Proposition~\ref{propnoise}, we begin with the following observation: from Proposition~\ref{incomp_prop}(e) and \ref{incomp_prop}(a) it follows that
\begin{align}
I_{\rm max}({\bf a}) &:=\sup_{\M,\N} I_{\bf a}(\M,\N)=\sup_{\theta\in (0,\pi)} I_{\bf a}(\P_0,\P_\theta)\\
& =f_a^{-1}\Big(\sup_{\theta\in (0,\pi)} I_b^{\rm noise}(\P_0,\P_\theta)\Big),
\end{align}
where we have also used the fact that every effect is a convex combination of projections, and the monotonicity of $f_a$. Hence, maximal incompatibility can be determined from the special values $I_b^{\rm noise}(\P_0,\P_\theta)$ studied in the preceding section.
More specifically, we only need to investigate Eq.~\eqref{comp_eq}; the maximal value turns out to be the left root of the quadratic polynomial on the right-hand side. This gives $I^{\rm noise}_{b}(P_0,P_\theta)\leq I^{\rm noise}_{\rm max}(b)$, where
\begin{equation}\label{maxincomp}
I^{\rm noise}_{\rm max}(b):=\frac{1}{2+\sqrt{2(1-b^2)}}
\end{equation}
and this value is attained for the unique $\theta=\theta_b$ which fulfills $\cos\theta_b=\frac 12 [1-I^{\rm noise}_{\rm max}(b)]^{-2}-1$.
Using Proposition~\ref{incomp_prop}~(e), we now immediately obtain the following result:
\begin{proposition}\label{thm:spectrum}
For any pair of effects $(M,N)$ we have
\begin{equation}\label{bound}
I^{\rm noise}_b(M,N)\leq I^{\rm noise}_{\rm max}(b).
\end{equation}
If $M$ and $N$ are projections, then the equality 
\begin{equation}\label{eq:jb=tb}
I_{b}(M,N)=I^{\rm noise}_{\rm max}(b)
\end{equation}
holds if and only if the spectrum of the operator $\id-(M-N)^2$ contains the point
 \begin{equation}
\label{eq:eigenvaluemax}
\chi_b=\tfrac 14 (1-I^{\rm noise}_{\rm max}(b))^{-2} \, .
\end{equation}
\end{proposition}

Using the dual program \eqref{dual4} of the corresponding SDP \eqref{semidefmu}, and the fact that $I_b^\mathrm{noise}=f\circ I_{\bf a}$ (Proposition~\ref{propnoise}), we get from \eqref{bound} a tight inequality
$$
\frac{\langle \psi | (\mathbb B-\id)\psi\rangle}{\langle \psi | \id\otimes (\id+b B_2)\psi\rangle}\leq \frac{1}{1+\sqrt{2(1-b^2)}},
$$
for arbitrary choices of $\psi$, $A_1$, $A_2$, $B_1$, and $B_2$. Since the case $b=0$ reduces to Tsirelson's inequality, this can be regarded as a generalization of that well-known bound for quantum correlations.

For the qubit case $I^{\rm noise}_b(\P_0,\P_\theta)=I^{\rm noise}_{\rm max}(b)$ is attained for a specific $b^2$, provided that $\theta\leq \pi/2$. 
This is depicted in Fig.~\ref{fig:qubit}. If $\theta>\pi/2$, the maximum is never attained (see Fig.~\ref{fig:theta}). 
It is also instructive to reinterpret this via the following more general situation: we test whether the state of a quantum system is one of two given states $\varphi$ or $\psi$. Then $M=|\varphi\rangle\langle \varphi|$, $N=|\psi\rangle\langle \psi|$, so $MNM$ has only one eigenvalue $F^2:=|\langle \psi|\varphi\rangle|^2$, where $F$ is the fidelity. By Proposition~\ref{incomp_prop} (e) the corresponding angle $\theta_F=\arccos(2F^2-1)$ then determines the incompatibility $I_b^{\rm noise}(\M,\N)=I_b^{\rm noise}(\P_0,\P_{\theta_F})$. It is important to note that even though this depends only on the fidelity $F$ as expected, it is not monotonic in $F$; incompatibility does not measure distance between the vectors. This is evident in Fig.~\ref{fig:theta}: in the orthogonal case $F=0$ the measurements are compatible, and as $F^2$ increases to $1/2$, also $I_b^{\rm noise}$ increases. At a certain point $1/2\leq F^2\leq 1$, incompatibility starts to decrease (except in the discontinuous case $b=-1$), and compatibility holds again at perfect fidelity $F=1$.

In higher-dimensional problems the value of any incompatibility measure $I^{\rm noise}_b(\M,\N)$ at a pair of projections is determined as the supremum of $I^{\rm noise}_b(\P_0,\P_\theta)$ where $\theta$ takes all values in the spectrum of $NMN$ or, equivalently, $MNM$, except $0$ and $1$. This means that $\theta$ takes at most $\min \{\rank(M),\rank(N)\}$ values; this is the maximum number of different values of $b^2$ for which $I_b^{\rm noise}(M,N)$ can be maximal for a given pair of projections. Note that the number of different values of $\theta$ depends not only on the rank of the projections but also on the dimension of the ambient space; for instance if $\rank(M)=\rank(N)=2$, and the projections both act on three-dimensional space, then the intersection of the subspaces is necessarily spanned by one nonzero vector $\varphi$, which is therefore an eigenvector of $MNM$ with eigenvalue $1$, implying that there is room for only one $\theta$ value.

We now proceed to consider a systematic scheme of implementing different $\theta$ values using a restricted set of unitary operations.

\begin{figure}
\begin{center}
\includegraphics[scale=0.8]{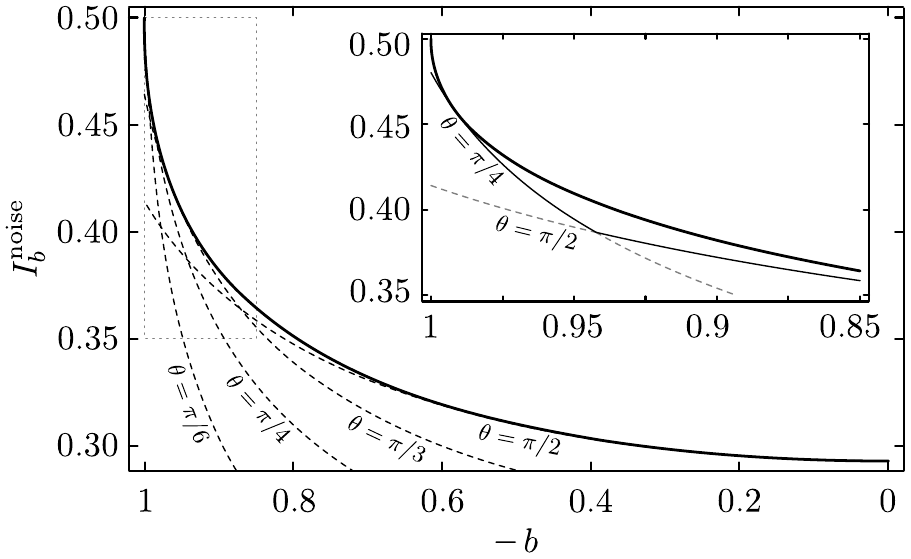}
\end{center}
\caption{\label{fig:qubit} The dependence of $I^{\rm noise}_b(\P_0,\P_\theta)$ on $b$ (on this plot shown negative $b$'s) for different choices of $\theta$ (dashed lines). Solid line represents the supremum over all choices of $\theta$, i.e., the $I^\mathrm{noise}_{\rm max}(b)$. The inset depicts the scaled area within the dotted rectangle for better visibility --- it depicts the situation of a four-dimensional case of Eq.~(\ref{circuitincomp}) and the following scheme for $n=2$. When the number of points $\theta$ in the spectrum (as well as $n$) increases, the supremum approaches the curve $I_{\rm max}(b)$.}
\end{figure}

\subsection{A quantum circuit implementation of maximal incompatibility}

In view of practical applications, it is crucial that quantum resources can be experimentally implemented. Clearly, any binary projective measurement can be realized by acting on a quantum state by a unitary operator and then measuring in a computational basis. In a realistic experiment, only certain unitaries (often called \emph{gates} in the context of quantum computation) can be actually implemented. Typical implementable gates are one-qubit rotations and two-qubit controlled rotations; they form a universal set which can be used to implement all unitary operations, and hence also all projective binary measurements, via suitable circuits. See, e.g., \cite{ZuEtal} for a recent work on experimental implementations.

Now, suppose that incompatibility appears as a resource in an experimental setting where only one-qubit rotations and two-qubit controlled rotations can be implemented. It is then crucial to know how maximal incompatibility (i.e., maximal noise robustness of the resource) can be achieved using these gates. The purpose of this subsection is to give an example of a circuit that does this, independently of the specifics of the experiment. A detailed and general study of efficiently implementing maximal incompatibility is beyond the scope of the paper.

Starting from a one-qubit system, the above pair $(P_0,P_\theta)$ can be understood as follows: one measures either directly in the computational basis, or performs first unitary quantum gate $U_\theta=\cos(\theta/2)\id+\sin(\theta/2)\sigma_y$. In this way the incompatibility $I_b^{\rm noise}(P_0,P_\theta)=I_b^{\rm noise}(P_0,U_\theta^*P_0U_\theta)$ can be thought of as being \emph{generated by a quantum gate}, and the same idea can of course be applied to more complicated systems. 
In fact, if $M$ and $N$ are arbitrary projections of the same dimension, we can always find a unitary $U$ such that $I_b^{\rm noise}(M,N)=I_b^{\rm noise}(M_0,U^*NU)$ where $M_0$ is diagonal in the computational basis.

Consider now an $n$-qubit system, with the basis measurement $M_0=\id^{\otimes n-1}\otimes |1\rangle\langle 1|$, i.e., only the last qubit is measured in the computational basis. In addition, assume that we can perform the Pauli-$x$ gate $\sigma_x$ on all qubits except the last one, and controlled rotations $CU_\theta$ which does $U_\theta$ on the $n$-th qubit if all the others are in state $|1\rangle$. 
Let $I_{n-1}$ denote the set of all binary sequences of length $n-1$, and for each ${\bf i}\in I_{n-1}$ define $W_{{\bf i}}(\theta)=X(i_1)\cdots X(i_{n-1})CU_{\theta}X(i_{n-1})\cdots X(i_1)$, where $X(i_k)$ is $\sigma_x$ on the $k$-th qubit if $i_k=1$, and identity otherwise. The Hilbert space decomposes into the direct sum $\hi=\oplus_{\bf i} \hi_{\bf i}$, where $\hi_{\bf i}={\rm span}\{|{\bf i}0\rangle, |{\bf i}1\rangle\}$, and $M_0=\oplus_{\bf i} |{\bf i}1\rangle\langle {\bf i}1|$, $W_{\bf i}(\theta)=\id\oplus\cdots\oplus\id\oplus U_{\theta}\oplus\id\oplus\cdots\id$, with $U_{\theta}$ in the ${\bf i}$-th block. We then choose for each ${\bf i}$ a value $\theta_{\bf i}\in [0,\pi/2]$, and set $W=\prod_{\bf i} W_{\bf i}(\theta_{\bf i})$. By Proposition~\ref{incomp_prop}(e), it follows that
\begin{equation}\label{circuitincomp}
I_{b}^{\rm noise}(M_0, W^*M_0W)=\max_{\bf i} I_b^{\rm noise}(P_0,U_{\theta_{\bf i}}^*P_0U_{\theta_{\bf i}}).
\end{equation}
If we choose all values $\theta_{\bf i}$ different, we can use this circuit to create a pair of measurements maximally incompatible for $2^{n-1}$ different choices of $b^2$. This situation is depicted in Fig.~\ref{fig:qubit}. 
As each of the curves touches the $I^{\rm noise}_{\rm max}(b)$ for a single $b^2$, their maximum can reach $I^{\rm noise}_{\rm max}(b)$ for only as many $b^2$'s as there are curves.

For instance, for $n=2$ we can insert two $\theta$ values, say $\theta_{0}=\pi/2$ and $\theta_1=\pi/4$. This circuit is depicted in Fig.~\ref{fig:circuit}, and the corresponding incompatibility is shown in Fig.~\ref{fig:qubit} (inset) --- the value $I^{\rm noise}_b(M_0,W^*M_0W)$ is for each $b$ by Eq.~(\ref{circuitincomp}) given as a maximum of the two incompatibility curves.
\begin{figure}
\begin{center}
\includegraphics[scale=0.9]{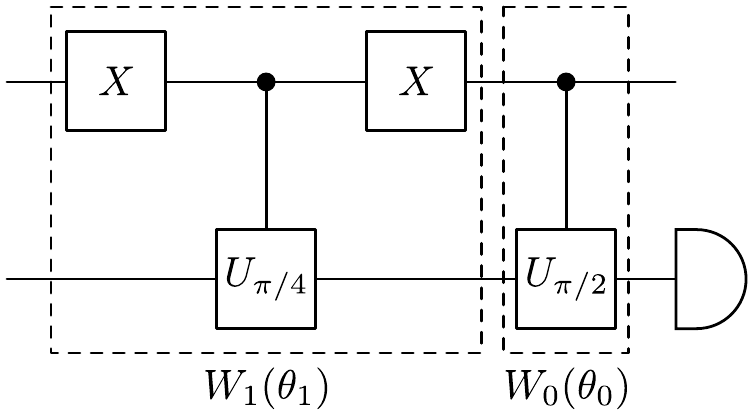}
\end{center}
\caption{\label{fig:circuit} A two-qubit quantum circuit that generates maximal incompatibility for two different monotones $I^{\rm noise}_{b_1}$ and $I^{\rm noise}_{b_2}$ where $\chi_{b_1}=\frac 12$ ($\theta=\pi/2$) and $\chi_{b_2}= \frac 12(1+1/\sqrt 2)$ ($\theta=\pi/4$).}
\end{figure}

Note that it is crucial that Alice's measurement at the end is only performed on the last qubit; this ensures that the projection $M_0$ is $2^{n-1}$ dimensional. 
By comparison, suppose that Alice measures all the qubits at the end, to check if the circuit produces a fixed binary sequence, say $|1,\ldots,1\rangle$. 
Then, regardless of the total circuit unitary $W$, the projections are just one dimensional, and we have a single $\theta$ value given by the fidelity $F=|\langle 1\cdots 1|U|1\cdots 1\rangle|$.

\begin{proposition}\label{limitprop}
With increasing number $n$ of qubits, the above quantum circuits can be used to produce binary measurements $(\M_n,\N_n)$ that are \emph{approximately maximally robust to noise} uniformly for any given bias $b$, i.e.,
$$
\lim_{n\rightarrow\infty} \max_{b\in [-1,1]} (I_{\rm max}^{\rm noise}(b)-I_{b}^{\rm noise}(\M_n,\N_n)) = 0 \, .
$$
\end{proposition}

\subsection{Maximal incompatibility of position and momentum}

It is clear from Proposition~\ref{limitprop} that truly maximally incompatible projections $M$ and $N$ can only exist in infinite dimensional systems, where the spectrum of $MNM$ fills the interval $[0,\frac 12]$. 
It is then natural to ask if such projections also have a physical meaning. 
Interestingly, this turns out to be the case: certain binarizations of the canonical variables $Q$ and $P$ for a one-mode continuous variable system have this property!

In order to see this, we split the real line $\real$ into positive and negative half lines.
This corresponds to asking if the result of $Q$ measurement is positive, and similarly for $P$-measurement.
Given that the wave function of the system is $\psi\in\ltwo{\real}$, the probabilities for the measurement outcomes are
\begin{align*}
\ip{\psi}{Q^+ \psi} &= \int_0^{\infty} \mo{\psi(x)}^2 dx,\\
\ip{\psi}{P^+\psi} &= \int_0^{\infty} \mo{\widehat{\psi}(p)}^2 dp,
\end{align*}
where $Q^+$ and $P^+$ denote the associated projections, and $\hat \psi$ is the Fourier transform. Using the fact that both projections are invariant under dilations, one can diagonalize them explicitly up to two-by-two matrices, as shown in \cite{KiWe10}. From the resulting decomposition it is then apparent that the spectrum of $Q^+P^+Q^+$ is the whole interval $[0,1]$. Hence we indeed have the following result.
\begin{proposition}\label{maxrob}
The half-line binarizations of position and momentum are maximally robust to noise, meaning that $I_{b}^{\rm noise}(Q^+,P^+)=I_{\rm max}^{\rm noise}(b)$ for all biases $b\in [-1,1]$.
Their incompatibility is more robust than any finite dimensional pair of binary measurements. 
\end{proposition}

We note that not all binarizations of position and momentum are maximally robust to noise.
In particular, a suitable periodic division of the real line $\real$ can make the binarizations even commutative \cite{KiWe10}, hence compatible already for $\lambda=0$.

\section{A quantum incompatibility game}
\label{sec:game}

The usefulness of quantum resources is sometimes analyzed via a game between two opponents,  quantum physicist (QP) and local realist (LR); see, e.g., \cite{DaGiGr05}. 
Here we provide a simple example of such a game, in which incompatibility is the quantum resource, and the quantity $I_b^{\rm noise}$ is the relevant figure of merit. 
Since the operational context is clearest in the CHSH experiment already considered above, we restrict to that scenario.

The challenge of the game is that QP has to design an experimental situation leading to a measurement outcome distribution for which the correlations between Alice and Bob are nonclassical in the sense that Bell inequality is violated. The experiment must be local in the sense that classical communication between the two parties is forbidden. 
Relying on quantum physics, the two resources that QP necessarily needs in order to win the game are (a) a source of entangled states and (b) a collection of local incompatible measurements for both parties. 
If the states and measurements are appropriately chosen then QP can violate the Bell inequality, thereby winning the challenge.

If we assume that resource (b) is unrestrictedly available, the relevant figures of merit are those quantifying resource (a). 
According to the general idea described earlier, we wish to investigate the opposite, assuming that resource (a) is not an issue, while resource (b) is restricted. 
We look at the situation from Alice's point of view, assuming that Bob has unrestricted resources.

We can think of LR as the ``evil'' Eve component in the scheme, disturbing Alice, effectively causing some noise in her measurements. The task for QP is then \emph{to choose a pair of incompatible quantum measurements that is most robust to noise}, so that Bell inequality is violated despite Eve's interference. It follows from the above development that the quantity $I_b^{\rm noise}(\M,\N)$ tells the amount of $b$-biased noise that LR needs to add so as to destroy any Bell violations, assuming Alice's measurements are $\M$ and $\N$.

There are now different scenarios depending on how much control on the noise LR is assumed to have. 
Each of these illustrates different aspects of the earlier theory. 
As before, the noise parameters are $(\lambda, b)$. 
We let $\lambda_{\rm LR}$ denote the maximal amount of noise LR can add. 
(In a real scenario, this could be related to, e.g., the duration of the measurement.)

\paragraph{LR-controlled bias.}
Assuming that LR has control on the bias $b$ of the noise, her optimal strategy is clearly to choose $b$ that minimizes $I_b^{\rm noise}(\M,\N)$, for a given choice $\M,\N$ of QP. This in turn implies that the relevant figure of merit for QP is the corresponding amount of noise $\inf_b I_b^{\rm noise}(\M,\N)=1-j(\M,\N)$. 
This means that QP must choose $\M,\N$ for which $j(\M,\N)$ is minimal. 
In our case of binary measurements, we simply have $j=1-I_{b=0}^{\rm noise}$, that is, the minimum point $b=0$ is independent of $(\M,\N)$. Hence, \emph{the optimal strategy for $LR$ does not depend on the choice of QP}. Assuming QP is restricted to projective measurements, we get from \eqref{Inolla}, an explicit expression
\[
j(M,N)=(1+2\no{[M,N]})^{-\frac 12} \, .
\]
Hence, QP should choose $M,N$ such that $\no{[M,N]}=1/2$, so that the Tsirelson's bound is achieved. 
Thus, \emph{assuming optimal strategy for LR, the optimal strategy for QP is fixed.} Then LR wins exactly when $\lambda_{\rm LR}\geq 1-1/\sqrt 2$. It is important to note that the optimal strategy for QP can already be realized by qubit measurements.

\paragraph{Fixed bias known to QP.}
Let us now assume that LR has no control over the bias parameter $b$, which is held fixed (e.g., by the construction of the measurement devices). While the strategy of LR is trivial in this scenario, it turns out that from the point of view of QP, the challenge becomes more interesting. Assuming that QP knows the bias $b$, he should choose a pair $\M,\N$ with the spectrum of $\id-(N-M)^2$ containing the point $\chi_b$ of \eqref{eq:eigenvaluemax}, so that the amount of noise LR has to add in order to destroy Bell violations is maximal, $I_b^{\rm noise}(\M,\N)=I^{\rm noise}_{\rm max}(b)$, for this particular bias $b$ (see Fig.~\ref{fig:theta}). Note that this strategy can be realised with qubit measurements. Hence, LR wins exactly when $\lambda_{\rm LR}> I^{\rm noise}_{\rm max}(b)$.

\paragraph{QP-controlled bias.}
We may also consider the case where QP can control the bias, by e.g.~selecting a measurement device with known systematic error. Now QP should choose the bias and the measurements $\M,\N$ such that $I_b^{\rm noise}$ is as large as possible. 
From Fig.~\ref{fig:qubit} it is clear that the optimal choice is not the unbiased case where noise robustness is restricted by the Tsirelson's bound. 
In fact,  \emph{destroying incompatibility is more difficult with strongly unbiased noise}. Hence, QP should choose the \emph{maximally unbiased case} $b=\pm 1$, and measurements close to being commutative. Then LR needs $\lambda_{\rm LR}\geq 1/2$ to win. As mentioned above, this amount of noise is enough to destroy incompatibility of \emph{any} pair of POVMs with arbitrary number of outcomes.

\paragraph{Fixed but unknown bias.} Here the bias is assumed to be fixed, but unknown to QP. Hence QP may assume it to be drawn randomly from the uniform distribution \footnote{We could also draw $b^2$ instead of $b$. While this would change the resulting probabilities, it does not affect the optimal strategy.}. 
Now the \emph{optimal strategy for QP is given by the value of $\theta$ which minimizes the probability that LR wins.} Assuming that QP knows $\lambda_{\rm LR}$, he can determine this probability: 
\begin{equation*}
P_{\rm LR, win}(\lambda_{\rm LR},\theta)={\rm Prob}(I_b^{\rm noise}(\theta)\leq \lambda_{\rm LR})=b_{\lambda_{\rm LR}}(\theta) \, , 
\end{equation*}
 where $b_\lambda(\theta)$ is determined by $\lambda=I^{\rm noise}_{b_\lambda(\theta)}(\P_0,\P_\theta)$.
It is clear from Fig.~\ref{fig:qubit} that the probability is minimized by choosing the value of $\theta$ for which $\lambda=I^{\rm noise}_{\rm max}(b_\lambda(\theta))$.
Hence, the optimal strategy for QP is to choose qubit measurements with $\cos\theta=\frac 12(1-\lambda_{\rm LR})^{-2}-1$ if $\lambda_{\rm LR}\geq 1-1/\sqrt 2$, and $\cos\theta=0$ (i.e., CHSH-optimal incompatibility) otherwise. In particular, if LR can cause more noise than required to destroy CHSH correlations, the optimal strategy for LR does \emph{not} involve CHSH-optimized measurements. 
With QP's optimal choice, she wins with probability 
\begin{align*}
P_{\rm QP, win}(\lambda_{\rm LR}) & =1-(I^{\rm noise}_{\rm max})^{-1}(\lambda_{\rm LR}) \\
&=1-\sqrt{\frac 12 \lambda_{\rm LR}^{-2}-(\lambda_{\rm LR}^{-1}-1)^2} \, , 
\end{align*}
 if $\lambda_{\rm LR}\geq 1-1/\sqrt 2$, and wins with certainty otherwise. Note that $P_{\rm QP, win}(1-1/\sqrt 2)=1$, and $P_{\rm QP, win}(1/2)=0$, as expected. The game is fair, i.e., $P_{\rm QP, win}(\lambda_{\rm LR})=\frac 12$, if and only if $\lambda_{\rm LR}=(2+\sqrt{3/2})^{-1}\approx 0.310$, which is only slightly larger than the minimal value $1-1/\sqrt{2}\approx 0.293$.

\paragraph{Fixed but unknown bias and magnitude.}
Here we also take $0<\lambda_{\rm LR}<\tfrac 12$ to be randomly chosen with uniform distribution. This case is interesting because the dimension of the available Hilbert space becomes relevant. Suppose first that only qubit resources are available. 
Then for measurements with angle $\theta$, the probability that QP wins is simply the probability that the randomly chosen point $(b,\lambda_{\rm LR})$ is under the curve $(b,I_{b}^{\rm noise}(\P_0,\P_\theta))$, $b\in [-1,1]$; see Fig.~\ref{fig:qubit}. 
Hence 
\begin{equation*}
P_{\rm QP, win}(\theta)=\int_{-1}^1 db \,I_{b}^{\rm noise}(\P_0,\P_\theta) \, , 
\end{equation*}
 and the optimal strategy can be computed by optimizing this function. 
 Now if we increase the available resources to include higher-dimensional measurements, QP's winning probability grows as more $\theta$ values can be included. 
The maximum possible probability is
$$
P_{\rm QP, win}=\int_{-1}^1 db \,I^{\rm noise}_{\rm max}(b)=\frac \pi2 (\sqrt 2-1)\approx 0.65 \, . 
$$
While the value itself is not of particular significance, the important point is the following:
 \emph{the optimal strategy requires maximal noise robustness in the sense of Proposition~\ref{maxrob}}. In particular, the Hilbert space must be infinite dimensional, and QP must choose a pair of projective measurements $(\M,\N)$ such that $\id-(N-M)^2$ has full spectrum, e.g., the binarizations of position and momentum.
 
\section{Summary and outlook} 

We have emphasized the role of incompatible measurements as a quantum resource necessary for, e.g., creating nonclassical correlations for the CHSH Bell scenario, by introducing the general notion of the \emph{incompatibility monotone}, and constructing a family of explicit instances for the case involving only binary measurements.

Similarly to quantum state resources, quantified by, e.g., entanglement monotones, there is no unique measure of incompatibility. Our choice $I_{\bf a}$, however, is motivated by several desirable properties: (a) $I_{\bf a}$ decreases under local operations, emphasizing its \emph{``dual'' nature to entanglement monotones}, and capturing the decay of incompatibility under noisy quantum evolution; (b) it \emph{operationally captures the local quantum resource} needed to violate CHSH Bell inequalities; (c) the special case $I_b^{\rm noise}$ has a \emph{direct operational meaning} as the magnitude of $b$-biased local noise needed to add to Alice's measurements so as to destroy all nonclassical CHSH correlations; and (d) $I_{\bf a}$ is \emph{computable via semidefinite program}, hence efficient for numerical investigation.

We have presented a detailed analysis of the properties of the noise-based monotone $I_b^{\rm noise}$, and provided an exemplary quantum circuit that could be used to implement maximally incompatible measurements in an experimental setting where only certain elementary gates are available. We have further illustrated the use of the quantity $I_b^{\rm noise}$, and its relationship to joint measurability degree, in the form of a quantum game, where one player aims to preserve the resource, and another one tries to destroy it via noise addition.

A couple of remarks concerning the specific nature of our setting are in order. First of all, the definition of incompatibility, and likewise the noise-robustness interpretation of our incompatibility monotones, make no reference to a possible tensor product structure of the underlying Hilbert space. In particular, the incompatibility resource does not have to be local; one can also investigate global incompatible measurements in a multipartite system: this would be relevant for, e.g., contextuality arguments as discussed in \cite{LiSpWi11}. However, in the above quantum game, as well as in the interpretation of the SDP monotones via the CHSH setting, incompatibility appears manifestly as a local resource, in contrast to the state resource, which is nonlocal. As a second remark, we emphasize that even though we restrict to the binary case, the idea of the incompatibility monotone is more general, and the present paper should be regarded as the first step towards a more complete picture. Our analysis shows that the problem is already nontrivial in the binary setting; in fact, the connection to Tsirelson's bound shows that it is as nontrivial as the structure of the CHSH inequality violations in the first place.

Further research in these directions will be generally aiming at clarifying the role of incompatibility at the ``measurement side'' of the quantum resource theory, dual to the ``state side,'' where massive efforts have been made to investigate entanglement and other forms of quantum correlations. In particular, it will be interesting to study specific quantum information protocols, where incompatibility monotones could serve as a useful figure of merit. For instance, one can investigate local aspects of decoherence in quantum control schemes, e.g.~involving specific sets of unitary operations used to create the measurements, and including environment-induced noise that gradually destroys the resource. Moreover, the connection to EPR steering requires further investigation.

\section*{Acknowledgments}
We acknowledge support from the Academy of Finland (Grant No.~138135), European CHIST-ERA/BMBF Project No.~CQC, EPSRC Project No.~EP/J009776/1, Slovak Research and Development Agency grant APVV-0808-12 QIMABOS, VEGA Grant No.~QWIN 2/0151/15 and Program No.~SASPRO QWIN 0055/01/01. The authors thank R. Uola, and one of the anonymous referees for useful comments on the manuscript.

\appendix

\section{Proof of Proposition~\ref{incomp_prop}}
\label{app:proof}

There exists $G_{ij}$ fulfilling \eqref{equalityconst} and \eqref{semidefmu} for $\mu=0$ if and only if $M$ and $N$ are compatible. Moreover, $I_{\bf a}(M,N)=I_{\bf a}(N,M)$ because the interchange of $N$ and $M$ corresponds to the interchange $G_{ij}\mapsto G_{ji}$ in the equality constraints, and this leaves the semidefinite constraints \eqref{semidefmu} unchanged since $(a_{ij})$ is symmetric. In order to show monotonicity, we let $\Lambda$ be a unital completely positive map, and suppose that $\mu\geq 0$ is admissible for the pair $(N,M)$, with $(G_{ij})$ the associated operators satisfying the constraints. Then by linearity, positivity and unitality, $(\Lambda(G_{ij}))$ satisfies the same constraints with $M$ and $N$ replaced by $\Lambda(M)$ and $\Lambda(N)$, respectively, i.e., $\mu$ is also admissible for $(\Lambda(N), \Lambda(M))$. This implies that $I_{\bf a}(N,M)\geq I_{\bf a}(\Lambda(N),\Lambda(M))$. Hence $I_{\bf a}$ is an incompatibility monotone.

In order to prove (a), we let $N,M_n\in\mathcal{E}$, take $\lambda_n\geq 0$ with $\sum_n\lambda_n=1$, and put $M:=\sum_n\lambda_n M_n$. Suppose that $\mu$ is admissible for each pair $(M_n,N)$, and let $(G^{(n)}_{ij})$ be corresponding operators satisfying the constraints. Then $\mu$ is admissible for $(N,M)$ because the operators $G_{ij}=\sum_n\lambda_n G^{(n)}_{ij}$ satisfy the constraints. Since the set of admissible values is always a half line, this implies (a). 

In order to prove (b), we decompose $\mathcal H$ now into an appropriate direct sum, and let $M=\oplus_n M_n$, $N=\oplus_n N_n$ (meaning that, e.g., $M_n$ is supported inside the subspace $\mathcal H_n$). Assuming first that $\mu$ is admissible for $(M_n,N_n)$ for all $n$, with $G^{(n)}_{ij}$ the corresponding operators on $\mathcal H_n$, we see $\mu$ is also admissible for $(M,N)$, with $G_{ij}:=\oplus_n G_{ij}^{(n)}$. Conversely, suppose that $\mu$ is admissible for $(M,N)$, with operators $G_{ij}$. Then $\mu$ is admissible for $(M_n,N_n)$ with the operators $G_{ij}^{(n)}:=P_nG_{ij}P_n$ because $N_n=P_nNP_n$ and $M_n=P_nMP_n$, and $P_n$ is the identity operator on the subspace $\mathcal H_n$. Hence (b) holds.

Part (c) follows from (b) and unitary invariance. Part (d) follows directly from the definition.

Concerning (e), each pair of projections can be diagonalized simultaneously up to two-by-two blocks if the Hilbert space is finite dimensional \cite{Ha69, RaSi89}; in a suitable basis the Hilbert space decomposes into a direct sum
\begin{align*}
\mathcal H &=\mathcal{M}_c\oplus \left(\oplus_n \mathbb C^2\right),\\
M &= M_0\oplus\left(\oplus_n P_0\right), & N&=N_0\oplus\left(\oplus_n P_{\theta_n}\right),
\end{align*}
where $\tfrac 12(1+\cos\theta_n)$ are the eigenvalues different from $0$ and $1$ of the operator $\id-(M-N)^2$, appearing according to their multiplicities, and $[N_0,M_0]=0$. In the infinite-dimensional case the spectrum may also be continuous, the direct sum is replaced by a direct integral, and the statement holds with ``set of eigenvalues'' replaced by ``spectrum.'' Using (b) and (c) we thus get (e).


\section{SDP duality and CHSH expression}
\label{app:sdp}

The dual to the semidefinite program (\ref{program}) can be computed in a straightforward fashion; the result is
\begin{equation}\label{dual2}
I_{\bf a}(M,N)=\sup_{X,Y,Z} {\rm tr}[(N-\id)Y]-{\rm tr}[MX]+{\rm tr}[(M-N)Z],
\end{equation}
with constraints
\begin{align*}
&X\geq 0,\, Y\geq 0, \, Z\geq 0\\
&Z\leq X+Y,\\
&\tr{a_{00}Y+a_{01}Z+a_{10}(X+Y-Z)+a_{11}X}=1.
\end{align*}
Since $R:=X+Y$ is a positive operator, it can be inverted inside its support projection. The conditions $Y\leq R$ and $R-Z\leq R$ are then equivalent to $Q:=R^{-\tfrac 12}YR^{-\tfrac 12}\leq\id$ and $P:=R^{-\tfrac 12}(R-Z)R^{-\tfrac 12}\leq\id$. Following \cite{WoPeFe09}, we then define $\Omega:=\sum_i |ii\rangle\in \mathcal H \otimes \mathcal H$ in a fixed basis, and let $A\mapsto A^{\rm T}$ denote the transpose in this basis. Then for any matrices $A,B$ we have
$\langle \Omega |(A\otimes B)\Omega\rangle ={\rm tr}[A B^{\rm T}]$.
Putting $\psi:=(\id \otimes (\sqrt{R})^{\rm T})\Omega$ we have
\begin{align*}
{\rm tr}[MX]&=\langle \psi| M\otimes(\id-Q^{\rm T})\psi\rangle,\\
{\rm tr}[(N-\id)Y]&=\langle \psi| (N-\id)\otimes Q^{\rm T}\psi\rangle,\\
{\rm tr}[(M-N)Z]&=\langle \psi| (M-N)\otimes(\id-P^{\rm T})\psi\rangle.\\
\end{align*}
and setting
\begin{align*}
A_1&:=\id-2N,  & A_2 &:=2M-\id,\\
B_1&:=\id-2P^{\rm T},  & B_2 &:=\id -2Q^{\rm T}
\end{align*}
after simple algebraic manipulations we get \eqref{dual4}.

\section{Qubit incompatibility}
\label{app:qubit}

We use the known characterization of compatibility of binary qubit POVMs \cite{Busch} (see also \cite{TeikoDaniel} and \cite{YLLO}): For two effects $E=\frac 12(\alpha \id+{\bf m\cdot \sigma})$ and $F=\frac 12(\beta \id+{\bf n\cdot \sigma})$ define $\langle E|F\rangle:=\alpha\beta-{\bf m\cdot n}$. Denote, e.g., $E^\perp :=\id-E$. Then $E$ and $F$ are compatible if and only if
\begin{align}\label{coexfunc}
& \left[ \langle E|E\rangle\langle E^\perp|E^\perp\rangle\langle F|F\rangle\langle F^\perp|F^\perp\rangle \right]^{1/2}-\langle E|E^\perp\rangle\langle F|F^\perp\rangle\nonumber\\
&+\langle E|F^\perp\rangle\langle E^\perp|F\rangle+\langle E|F\rangle\langle E^\perp|F^\perp\rangle \geq 0.
\end{align}
As the expression shows, deciding incompatibility is already considerably nontrivial in the qubit case. We look at effects $E=M_{\lambda,b}$ and $F=N_{\lambda,b}$ defined by \eqref{eq:noisym} via the two projections $M=P_0$ and $N=P_\theta$. After some algebraic manipulation, \eqref{coexfunc} reduces to
$$
[(1-\lambda)^2\cos\theta-\lambda^2b^2]^2-2(1-\lambda)^2+1-2\lambda^2b^2 \geq 0.
$$
The left hand side changes sign at the unique point $\lambda=I_{b}^{\rm noise}(P_0,P_\theta)$, which therefore solves the resulting equation \eqref{comp_eq}.

\section{Monotonicity of $b\mapsto I_b^{\rm noise}$}
\label{app:monotonicity}

Given two projective measurements $\M$ and $\N$ we will now show that $I_b^{\rm noise}(\M,\N)$ is a non-decreasing function of $|b|$. Since $I_b^{\rm noise}(\M,\N)$ is given as a supremum of the values $I_b^{\rm noise}(\P_0,\P_\theta)$ where $\theta$ varies over a subset of $[0,\pi]$, it suffices to show that $I_b^{\rm noise}(\P_0,\P_\theta)$ is non-decreasing in $|b|$ for each $\theta$. This is apparent in Fig.~\ref{fig:theta}, and can be shown analytically using \eqref{comp_eq} as follows: We fix $0\leq \theta\leq \pi$ in the following. Defining
\[
f_\theta(\lambda,b)=(1-\lambda)^2\cos\theta-\lambda^2b^2
\]
for $0\leq \lambda \leq 1/2$ and $-1\leq b\leq 1$, we can rewrite Eq.~(\ref{comp_eq}) as
\begin{equation}
\label{comp_eq_f}
f_\theta^2+2f_\pi+1=0.
\end{equation}
For shortness we do not explicitly write the dependence on $\lambda$ and $b$ from now on in $f_\theta$ and its derivatives. 
This implicitly determines $\lambda=\lambda(b)$ as a function of $b$, and we can find $\lambda':=d\lambda/db$ by differentiating:
\[
f_\theta \frac{df_\theta}{db}+\frac{df_\pi}{db}=0,
\]
where
\[
\frac{df_\theta}{db}=-\frac{2\lambda'}{1-\lambda}\left(f_\theta+\lambda b^2\right)-2b\lambda^2.
\]
Solving for $\lambda'$ we find
\[
\lambda'=-\frac{\lambda^2(1-\lambda)b(f_\theta+1)}{f_\theta^2+f_\theta\lambda b^2+f_\pi+\lambda b^2}.
\]
For the numerator we have $\lambda^2(1-\lambda)\geq 0$ and
\[
f_\theta(\lambda,b)+1\geq f_\pi(\lambda,b)+1=2\lambda-\lambda^2(1-b^2)\geq0
\]
for all relevant $\lambda$ and $b$. Using Eq.~(\ref{comp_eq_f}) to get rid of the second power of $f_\theta$, we find that the denominator equals $\lambda D$ where
\begin{align*}
D&=b^2[(1-\lambda)^2\cos\theta-\lambda^2 b^2+1]+\lambda(1+b^2)-2\\
&\leq(1-b)^2(\lambda^2+\lambda-2)\leq 0
\end{align*}
with equality only for the boundary case $b^2=1$ and $\theta=0$. Hence, we find that
\[
\mathrm{sgn}(\lambda')=\mathrm{sgn}(b)
\]
with $\lambda'=0$ only when $\cos^2\theta=1$ (a boundary case) or for $b=0$.
Hence, $\lambda(b)=I_b^{\rm noise}(\P_0,\P_\theta)$ is nondecreasing in $|b|$.

\end{document}